\pdfoutput=1
\documentclass[12pt,psamsfonts]{amsart}
\usepackage{amsmath,amsthm,amsfonts,amssymb}
\usepackage{eucal,amscd}
\usepackage{graphicx}

\input epsf
\pdfoutput=1

\numberwithin{equation}{section}

\newtheorem{theorem}[equation]{Theorem}

\newtheorem{lemma}[equation]{Lemma}

\theoremstyle{definition}

\newtheorem{defin}{Definition}
\newtheorem{question}{Question}
\newtheorem{proposition}[equation]{Proposition}

\begin{document}

\title[Approximate Counting and Quantum Computation]
{Approximate Counting and Quantum Computation}

\author{M. Bordewich}
\email{magnusb@comp.leeds.ac.uk}
\address{Mathematical Institute\\University of Oxford\\ 24-29 St. Giles'\\Oxford, OX1}

\author{Michael Freedman}
\email{michaelf@microsoft.com}
\address{Microsoft Research\\One Microsoft Way\\Redmond, WA 98052}

\author{L. Lov\'asz}
\email{lovasz@microsoft.com}
\address{Microsoft Research\\One Microsoft Way\\Redmond, WA 98052}

\author{D. Welsh}
\email{dwelsh@maths.ox.ac.uk}
\address{Mathematical Institute\\University of Oxford\\ 24-29 St. Giles'\\Oxford, OX1}

\begin{abstract}

Motivated by the result that an `approximate' evaluation
of the Jones polynomial of a braid at a $5^{th}$ root of unity can
be used to simulate the quantum part of any algorithm in the
quantum complexity class BQP, and results relating BQP to the
counting class GapP, we introduce a form of additive approximation
which can be used to simulate a function in BQP. We show that all
functions in the classes \#P and GapP have such an
approximation scheme under certain natural normalisations. However
we are unable to determine whether the particular functions we are
motivated by, such as the above evaluation of the Jones
polynomial, can be approximated in this way. We close with some
open problems motivated by this work.\end{abstract}

\maketitle

\section{Introduction}

The quantum complexity class BQP consists of those decision
problems that can be computed with bounded error, using quantum
resources, in polynomial time. Relative to the polynomial
hierarchy of classical computation, it is known that
$$\textrm{BPP}\subseteq \textrm{BQP}\subseteq \textrm{PP}\subseteq \textrm{PSPACE},$$ and at the
moment none of these inclusions is known to be
proper~\cite{adl97}. Recent work by Freedman, Kitaev, Larson and
Wang~\cite{fre02a} has shown that the `quantum part' of any
quantum computation can be replaced by an approximate evaluation
of the Jones polynomial of a related braid.   A classical
polynomial time algorithm can convert a quantum circuit for an
instance of such a problem, into a braid, such that the
probability that the output of the quantum computation is zero is
a simple (polynomial time) function of the Jones polynomial of the
braid at a 5$^{th}$ root of unity. For an exact statement of this
see Freedman, Kitaev, Larsen and Wang~\cite{fre02a}, or the more
detailed papers by Freedman, Kitaev and Wang~\cite{fre02b}, and
Freedman, Larsen and Wang~\cite{fre02c, fre02d}.

It therefore follows that if we take $A(L,x)$ to be an oracle that
returns the evaluation of the Jones polynomial of a braid $L$ at a
point $x$, any BQP computation can be replicated by a classical
polynomial time algorithm with one call to $A$, i.e.
$\mathrm{BQP}\subseteq \mathrm{P}^{A}$. Since computing
the Jones polynomial is in general a \#P-hard problem, this does
not help. However, it is not an exact evaluation of the Jones
polynomial that is required, but an approximate evaluation at a
specific point for braids of a specific class. Hence we may look
for a weaker oracle $A'$ such that  $\mathrm{BQP}\subseteq
\mathrm{P}^{A'}$.

In a different approach Fortnow and Rogers~\cite{for99} link
quantum complexity to the classical complexity class GapP. In
particular they show that for any quantum Turing machine $M$
running in time $t(n)$ there is a GapP function $f$ such that for
all inputs $x$
$$\textbf{Pr}(M(x)\textrm{ accepts})=\frac{f(x)}{5^{2t(|x|)}}.$$
Again evaluating a general GapP function exactly is \#P-hard,
however one can simulate $M$ using a polynomial algorithm with
access to an oracle $A''$,  where $A''$ is an oracle giving an
approximation to the GapP function $f$.

With this motivation we examine the type of approximation needed
in order to simulate a quantum computation, and then consider the
complexity of such approximations. It turns out that an
\emph{additive approximation} is sufficiently powerful. We should
emphasize that a polynomial time additive approximation scheme is
weaker than the familiar and much studied \emph{fully polynomial
randomized approximation scheme} (FPRAS). However it is well known
that any function which counts objects for which the corresponding
decision problem is NP-complete cannot have an FPRAS (unless
NP=RP). We show below that \emph{all} \#P functions do have
polynomial time additive approximation schemes under natural
normalisations. We also show that in two senses this is the best
sort of approximation we can hope to achieve in polynomial time
(see Theorems~\ref{sample},~\ref{log} and \ref{bestnormcol}).


\section{Quantum computing}

A \emph{link}\index{link} $L$ is a smooth submanifold of $S^3$, consisting of $c(L)$ disjoint simple closed curves. A \emph{braid}\index{braid} on $m$ strings is constructed as follows. Take $m$ distinct points in a horizontal line ($p_1, p_2,\ldots,p_m$) and link them to $m$ distinct points ($q_1, q_2,\ldots,q_m$)  lying on a parallel line, by $m$ disjoint simple arcs $f_i$ in $\mathbb{R}^3$, so that $f_i$ starts at $p_i$ and ends at $q_{\pi(i)}$ where $\pi$ is a permutation. A braid can be closed in numerous ways, by identifying the points $p_i$ and $q_j$ in some way, creating a link. Similarly any link can be represented as a braid. In particular, the \emph{plat closure}\index{plat closure} of a braid on $2m$ strings is obtained by identifying the points $p_{2i-1}$ and $p_{2i}$, and $q_{2i-1}$ and $q_{2i}$ for $1\leqslant i\leqslant m$.

\subsection{Topological computing and the Jones polynomial}\index{topological quantum computing}\label{jones}

One of the major difficulties in building a quantum computer has been the sensitivity of the system to outside interference. Freedman, Kitaev, Larson and Wang~\cite{fre02a} introduced the notion of \textsl{topological quantum computing}, in an attempt to make the computations less sensitive to small disturbances. The basic idea is as follows. One can create pairs of special quasi-particles, called anyons, in a 2-dimensional plane sandwiched between two blocks of a superconductor. The anyons have a certain probability of annihilating each other (leaving a vacuum) when brought together. However this probability changes, according to the laws of quantum mechanics when one anyon is moved around the other before they are brought together. Even if it is moved in a complete circle around the other, on reaching its original position the probability of annihilation is changed. Thus a system of a large number of these particles can be used as a quantum computer for decision problems; pairs of anyons are created, moved around relative to each other, and then a predefined pair of the anyons is brought together. If this pair annihilate each other leaving a vacuum, this is taken to be an output of 0 (or rejection), if they do not it is taken to be an output of 1 (or acceptance). The paths in the 2-dimensional surface, combined with a time dimension, give rise to a 3-dimensional representation of the `computation' as a braid. There remain major difficulties in constructing such a quantum computer, and controlling the movement of anyons. However one of the important results of Freedman \textsl{et al.} is that small changes in the paths of the anyons do not affect the outcome of the computation, indeed it is determined by the isotopy class of the braid, and therefore stable under perturbations of the paths that do not change the braiding itself.

The way in which the probability changes is sufficiently subtle that such a quantum computer is universal in the following sense. The Kitaev-Solovay theorem~\cite{kit97,sol} together with the density theorem of Freedman, Larsen and Wang~\cite{fre02c, fre02d} yields an algorithm which given any quantum circuit on $m/2$ qubits and error parameter $\epsilon$, outputs a braid on $m$ strings using a polynomial number of crossings (polynomial in $m$ and $\log \epsilon^{-1}$). The topological quantum computation using this braid efficiently simulates the quantum circuit, (the probability of acceptance is within $\epsilon$ of the correct value). Since an algorithm for a BQP problem can be used to generate a quantum circuit for a given instance, the above result gives an explicit method for finding an equivalent topological quantum computation, and so the class BQP is the same under either model.

Hence a quantum computation on $m/2$ qubits is approximately represented by a braid $b$. In showing that the topological quantum computation depends on the isotopy class of $b$ alone~\cite{fre02a}, the following link $L$ is considered. $L$ is the plat closure of the composition of $b^{-1}, b$ and a small loop $\gamma$ inserted (between $b^{-1}$ and $b$) around the leftmost two strings (see Fig.~\ref{braidfig} and Fig.~\ref{linkfig}). Both $b$ and $b^{-1}$ are needed as any quantum computation must be reversible; the loop $\gamma$ effects a measurement of the qubit represented by the leftmost pair of strings. The conclusions of~\cite{fre02a} may then be summarized as the following theorem: refer to~\cite{fre02a} for full details.  %
\begin{figure}
\begin{center}
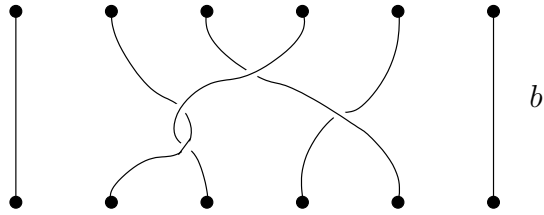%
 \caption{The braid $b$.} \label{braidfig}%
\end{center}%
 \end{figure}%
 \begin{figure}
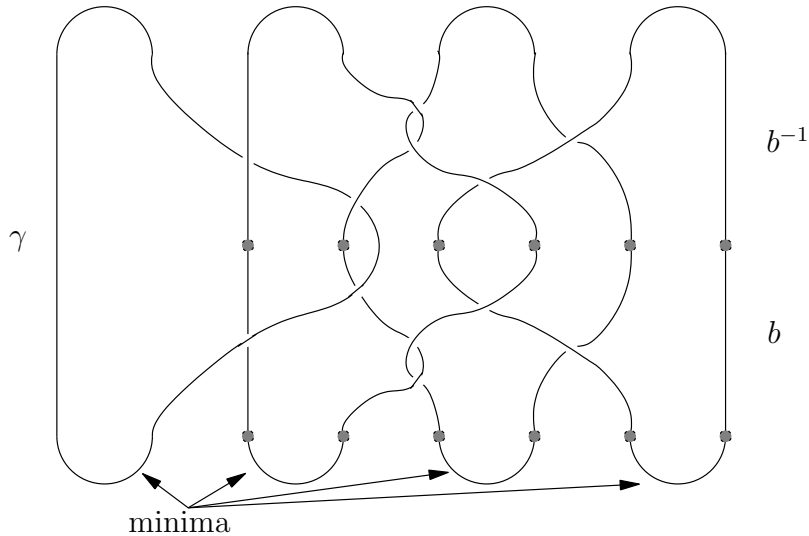%
 \caption{The link $L$.} \label{linkfig} %
\end{figure}%
\begin{theorem}\label{FKLWtheorem}
Let $\pi$ be a problem in BQP, with a polynomial time quantum algorithm $\mathcal{A}$, and let $\mathcal{I}$ be an instance of $\pi$. For any $\epsilon>0$, a link $L$ may be determined in time polynomial in $|\mathcal{I}|$ and $\log \epsilon$ such that
\begin{equation} \big|\emph{\textbf{Pr}}(\mathcal{A}(\mathcal{I})=0) - \frac{1}{1+[2]_5^2} \bigg( 1+
\frac{(-1)^{c(L)+w(L)}(-a)^{3w(L)} V_L(e^{2\pi i/5})}{
[2]_5^{m(L)-2}}\bigg)\big|< \epsilon,\label{FKLWeqn1}\end{equation}
 where
$a=e^{i\pi/10}$ and $[2]_5 =2 \cos \pi/5$ and $c(L), w(L)$ and $m(L)$
are the number of components, writhe and number of minima of
the link $L$ respectively. \end{theorem}
The minima of~\cite{fre02a} are the individual joins in the plat closure at the bottom of the braid, hence $m(L)$ is half the number of strings in $b$ plus one. By construction, the number of strings in $b$ is twice the number of (qu)bits in the input, hence $m(L)=|I|+1$. The writhe is defined for an oriented link, and is the number of `positively oriented' crossings minus the number of `negatively oriented' crossings, with respect to the given orientation of $L$. It is easily computable. The Jones polynomial is also defined for oriented links, however the formula above is independent of the orientation chosen for $L$. Since every crossing in $b$ appears reversed in $b^{-1}$, these do not contribute to the writhe of $L$, hence the writhe is determined by the four crossings involving $\gamma$, and can only be $-4, 0$,
or $4$ (depending on the orientation of the two strands passing
through $\gamma$). If $L^*$ denotes $L$ with the orientation of  one component reversed, then $V_{L^*}(t)=t^{-3\lambda/2}V_L(t)$, where $\lambda$ is the contribution to the writhe of $L$ from crossings of the reversed component over (or under) the rest of $L$. Hence only reversing $\gamma$ or one of the leftmost two strings can affect the Jones evaluation, and it is easily checked that this is compensated by the change in the terms involving the writhe. To be consistent, we will retain the notation $[2]_5$ for $2 \cos \pi/5$ from~\cite{fre02a} throughout.

It is Theorem~\ref{FKLWtheorem} that gives rise to our interest in approximating
$V_L(t)$. Further explanation of the derivation of this equation is given
in~\cite{wan03} where the special but sufficient case $w (L) =0$
is considered (we could restrict attention to braids with writhe zero without affecting the main results). Although this formula involves an evaluation at
$e^{2\pi i/5}$, similar results can be obtained for the $n^{th}$
root of unity for any $n \geq 5, n \neq 6$ but these involve
multiple $L$'s.

The preceding paragraphs explain how an evaluation of a Jones
polynomial can yield the answer to a (general) quantum
computation.  There is a weak converse to this.  Suppose we have a
quantum computer at our disposal with which to learn something
about a Jones evaluation of a link $L$.  We may assume without loss of generality
(Freedman, Larsen and Wang~\cite{fre02c}) that our quantum computer is of the
topological kind and thus nicely adapted to braids.  We can
(easily) write $L$ as  the plat closure of a braid $b$ by
starting with the link diagram and pulling the overcrossings up
and the undercrossings down.  Let $m$ be the number of strands of this braid $b$. If we wish to
evaluate $V_L (\alpha)$, $\alpha = e^{2 \pi i/r}$, we encounter an
important constant $d=2\cos \pi/r$.  The norm $|V_L (\alpha)|$ is
bounded from above by $d^{m/2}$ with $|V_L (\alpha)| = d^{m/2}$
 achieved only when $L$ is the unlink on $m/2$ components, a case which
occurs when $b$ is the identity braid.  Our quantum computer will be able to
provide an additive approximation (see below) of $|V_L (\alpha)|$ as a
variable with range $[0,d^{m/2}]$.

Given $m$ marked points in the horizontal plane and the number
$\alpha$, there is a finite dimensional Hilbert space $H$ on which
$m$-strand braids act through a Jones representation $p$. The
$m/2$ maxima (in the plat closure) determine a vector $c$ in this space and the $m/2$ minima (in the plat closure)
determine a vector in the dual $H^\ast$, which when identified
with $H$ by the hermitian inner product, is the same $c$.

We have
\begin{eqnarray*}
\frac{V_L(\alpha)}{d^{m/2}} = \langle c |p (b) |c\rangle.
\end{eqnarray*}
Furthermore Prob$(|0 \rangle) = |\langle c | p (b) | c
\rangle |^2$, where Prob$(|0 \rangle)$ refers to the physical
probability that below the cups, after all the `particles' have
been fused in pairs, the vacuum $|0\rangle$ is observed, that is no
nontrivial particles result from these fusions.  The last formula
reflects the quantum mechanical rule that the probability of
observing  an outcome, in this case $|0\rangle$, is proportional
to the square of the component of the state vector in the
$|0\rangle$-direction.

Because the range of $|V_L (\alpha ) |$ depends exponentially on the
number of strings in the braid, $m(b)$, our quantum computer will give much better information (sooner)
if we succeed in displaying $L$ with, or nearly with, the minimal
$m(b)$, called the \emph{braid index} of $L$.

Turning to the computational question, it is a theorem of Thistlethwaite~\cite{thi87} that when $L$ is an
alternating link, with associated plane graph $G$, then \begin{displaymath}
V_L(t)=\alpha T(G;-t,-t^{-1})\end{displaymath}
where $\alpha$ is an easily computable function, and $T$ is the Tutte
polynomial of the planar graph.
It is known \cite{ver92} that even for planar graphs, computing
$T(G;x,y)$ is \#P-hard, except when $(x,y)$ is one of a few
special points, or lies on a hyperbola satisfying
$(x-1)(y-1)=q\in\{1,2\}$.

Since $(-e^{2\pi i /n},-e^{-2\pi i /n}), n \geqslant 5,$ is not one of these
`easy' points, exact computation in polynomial time is not
feasible (unless \#P=P). It also seems unlikely that an FPRAS
exists for these points. However the notion of an FPRAS seems to
be much stronger than the kind of approximation that is needed in
the current context. For any BQP language $L$ there is a quantum
Turing machine $M$ such that for all $x\in L$, $M$ accepts with
probability at least $3/4$, and for all $x\not\in L$, $M$ accepts
with probability at most $1/4$. Therefore all that we require is
to determine which quartile of its range $V_L(e^{2\pi i/5})$ lies
in. We return to this topic in Section~\ref{TP}.

When considering algorithms on braids or links, the size of the input is taken to be the number of crossings. A quantum gate on $m/2$ qubits is converted into a braid on $m$ strings of length polylog $(1/\epsilon)$, therefore the number of crossings in the braid associated with a BQP circuit is polynomially related to the
number of gates in the BQP circuit. This in turn is bounded by a
polynomial in the input size, hence an algorithm will either be
polynomial with respect to both the number of crossings in the braid and the number of input qubits to the circuit, or neither.

\subsection{GapP functions}
The class of counting functions which constitute \#P is the set of
functions that count certificates of membership of a language
belonging to NP, hence \#P functions are constrained to evaluate
to non-negative integers. The class of functions GapP can be
regarded as the closure of \#P under subtraction, that is to say
a function $f:\mathcal{I}\mapsto \mathbb{Z}$ is in GapP if and
only if there exist functions $g,h \in \#P$ such that
$f(I)=g(I)-h(I)$ for all $I\in \mathcal{I}$. The class AWPP can be
defined as follows~\cite{fen03}. A language $L$ is in AWPP if and
only if there exist a polynomial $p$ and a GapP function $g$ such
that for all $I\in \mathcal{I}$,
$$I\in L \quad \Rightarrow \quad \frac{3}4 \leqslant \frac{g(I)}{2^{p(|I|)}} \leqslant 1\phantom{n}$$
\begin{equation}I\not\in L \quad \Rightarrow \quad 0 \leqslant \frac{g(I)}{2^{p(|I|)}} \leqslant \frac{1}4.\label{AWPP}\end{equation}

The increase in power of quantum computation over classical
computation is that in a quantum computer there is an ability to cancel out computations paths.
Fortnow and Rogers~\cite{for99} show that this power is captured
by the class GapP, in which a similar effect is seen. In
particular they show that BQP$\subseteq$AWPP. It therefore follows
that for a BQP language $L$, polynomial $p$ and  GapP function $g$
satisfying (\ref{AWPP}), determining which quartile of the range
$[0,2^{p(|I|)}]$ contains $g(I)$ would be enough to determine membership of
$L$.

To summarize, our foremost problems can be interpreted as finding
a suitable approximation for the Jones polynomial of a link,
$V_L(t)$, the Tutte polynomial of an associated planar graph,
$T(G;x,y)$, at a particular point, or for the GapP functions
arising from BQP languages.


\section{Approximation}
Given a function $\psi:\mathcal{I}\mapsto \mathbb{R}$ for which no efficient exact evaluation algorithm is known, one may be interested in an `approximate' answer instead. A standard approach is to look for a fully polynomial randomized approximation scheme (FPRAS) for the problem. If $\psi$ is such a function and $I\in \mathcal{I}$ is an input, then an FPRAS for $\psi$ is a randomized algorithm that given any $I\in\mathcal{I}, \epsilon>0$ will output  $\hat{\psi}(I,\epsilon)$, such that
\begin{displaymath}
\textbf{Pr}[|\hat{\psi}(I,\epsilon)-\psi(I)|>\epsilon\psi(I)]< 1/4,\end{displaymath}
and the running time is polynomial in $|I|$ and $\epsilon^{-1}$.

Here one might be prompted to consider the following sort of approximation: suppose we know a range in which the answer lies; can we say where in that range the answer lies? Is it in the top or bottom half of the range, or in which quartile? We shall see that this approach is unlikely to be feasible, and in Section \ref{AAsect} we present an alternative. Clearly this type of approximation depends on the nature of the range. For the moment let us restrict our attention to the class of functions in \#P. We will make the standard assumption that for a given NDTM $M$ there exists a fixed polynomial $p$ such that for any input $x$,  all certificates have size $p(|x|)$ (so the total number of possible certificates of $M$ is $2^{p(|x|)}$). We would like to answer the following problem, denoted  by $\pi_r$: given $r$, for which $k$ is the number of accepting certificates for $x$ between $\frac{(k-1)}{r}2^{p(|x|)}$ and $\frac{k}{r}2^{p(|x|)}$?

The problem $\pi_2$ is simply to determine which inputs have more than half of all certificates as accepting certificates. The set of languages in this class is exactly the set PP of probabilistic polynomial time languages. Furthermore, $\pi_2$ is clearly Turing reducible to $\pi_{2s}$, for any positive integer $s$, since if $\pi_{2s}(x)\leq s$ then $\pi_2(x)=1$, otherwise $\pi_2(x)=2$. Hence it is no surprise that this approach to approximation is NP-hard for \#SAT, indeed the following lemma shows that any attempt to approximate \#SAT in this way, or any problem with a parsimonious reduction to SAT, is unlikely to work. The proof is straightforward and we omit it, details may be found in~\cite{bor03c}.

\begin{lemma} \label{PP} For $k\in \mathbb{Z}$, deciding whether a CNF formula in $n$ literals has more than $2^{n-k}$ solutions is NP-hard. \end{lemma}

When $k=1$ the same decision for disjunctive normal form (DNF) formulae is equivalent to that for SAT, since the negation of a SAT formula is in DNF, and hence for an instance $F$, we have \#SAT$(F)=2^n- $\#DNF$(\overline{F})$, where $n$ is the number of literals. This observation leads to the following related lemma. Again, the proof is omitted and details may be found in~\cite{bor03c}.
\begin{lemma}For $k\in \mathbb{Z}$, deciding whether a DNF formula in $n$ literals has at least $2^{n-k}$ solutions is NP-hard. \end{lemma}
This may seem more counterintuitive since not only is DNF in P, but also \#DNF has an FPRAS~\cite{kar89}. On the other hand, the next lemma shows that the number of  stable (independent) sets of vertices in a graph (\#SS) can be approximated in this way, even though it is \#P-complete and does not admit an FPRAS unless NP=RP. Essentially this is because the `natural' upper bound on the number of stable sets, $2^n$, is far too big unless the graph has very few edges. For details of the proof see~\cite{bor03c}.

\begin{lemma}  Let $G$ be a graph on $n$ vertices. For $r\in\mathbb{Z}$, determining for which $k$, $\#SS(G) \in \left[\frac{(k-1)}{r}2^{n},\frac{k}{r}2^{n}\right)$ is computable in time polynomial in $n$ and $r$. \end{lemma}

Lemma~\ref{PP} suggests that we cannot hope to fix a partition of the range and then determine in polynomial time in which section the answer lies; the difficulty associated with an NP-complete decision problem can be shifted to exactly the boundary between two parts of our partition. We therefore consider an alternative method of approximation which will meet our needs.

\subsection{Additive approximation}\label{AAsect}

Our approach to approximation consists of determining a small section of the range depending on the input, and in which we can say the answer lies with high probability. This gives rise to an additive approximation.

\begin{defin}\textbf{Additive Approximation (AA):}\label{AAdef}\\
Given any function $f:\mathcal{I}\mapsto\mathbb{C}$ and a normalisation $u:\mathbb{Z}^+\mapsto\mathbb{R}^+$, an additive approximation for $(f, u)$ is a probabilistic algorithm which given any $I\in\mathcal{I}, \epsilon>0$ produces an output $\hat{f}(I)$, such that
$$ \mathbf{Pr}[|f(I)-\hat{f}(I)|>\epsilon u(|I|)]<1/4,$$
in time polynomial in $|I|$ and $\epsilon^{-1}$.
\end{defin}
Note that the $1/4$ in the definition could be replaced by any $\delta\in(0,1/2)$, since we could reduce this error probability in polynomial time by taking several runs of the algorithm. Note also that most of the time we shall be considering the case where $f$ is real. In contrast to
the set of functions admitting an FPRAS, which is closed under
addition but not under subtraction (e.g. \#DNF$(f)$ has an FPRAS,
but \#SAT$(f)=2^n-\#$DNF$(\bar{f})$ does not), we have the following result whose proof we leave to the reader.
\begin{proposition}\label{add} Suppose $(f,u)$ and $(g,v)$ admit AA
algorithms, then there
exists AA algorithms for  $(-f,u)$, $(f+g,u+v)$ and $(f-g,u+v)$. If, in addition, $|f(I)|\leqslant u(|I|)$ and $|g(I)|\leqslant v(|I|)$ for all $I$, then  there is an AA algorithm for $(fg,uv)$.\end{proposition}

The normalisation is  crucial. Since we are most interested in determining where in the range of possible values the answer lies, we shall usually be taking $u$ to be an upper bound on $|f|$ depending only on input size. An additive approximation allows errors up to an absolute value of $\epsilon u(|I|)$, whereas an FPRAS allows only errors up to an absolute value of $\epsilon f(I)$. It is therefore a weaker notion of approximation, and it is easy to check that any function that admits an FPRAS also admits an AA algorithm under any upper bound.
\begin{lemma} Let $f:\mathcal{I}\mapsto \mathbb{R}$ be a function that admits an FPRAS, and let $u:\mathbb{Z}^{+}\mapsto\mathbb{R}$ satisfy $|f(I)|\leqslant u(|I|)$ for all inputs $I\in \mathcal{I}$. Then $(f,u)$ has an AA algorithm.\end{lemma}

Note also that a given function will have an AA with respect to some normalisations but not others. For example we show later that for the number of proper three colourings of a connected graph $G$ on $n$ vertices, $P_G(3)$ where $P_G$ is the chromatic polynomial of $G$, we have an AA for $(P_G(3),2^{n})$. However for any constant $\delta>0$, $(P_G(3),(2-\delta)^{n})$, does not have an AA unless NP=RP (Theorem~\ref{bestnormcol}). In other words we can determine $P_G(3)$ to within an additive error $\epsilon 2^n$ in polynomial time, but we cannot approximate to within an additive error $\epsilon (2-\delta)^n$. Note that if $(f(I),u(|I|))$ has an AA, then for any fixed polynomial $p$, $(f(I),u(|I|)/p(|I|))$ also does, since we can absorb the polynomial factor in the normalisation into $\epsilon$ at only a polynomial slowing of the algorithm.

It is the determination of the `best' normalisation for a given function that causes the
greatest difficulties, particularly in relation to approximating
$V_L(t)$. Nevertheless our first positive result shows that any
function belonging to \#P does have an AA algorithm under very natural
 normalisations.


\section{Additive approximations for \#P functions}
The class of functions which constitute \#P can be regarded
as the set of functions that count certificates of membership of a
language belonging to NP. For a given NP-language $L$ there will be
infinitely many NDTM's which check membership of $L$, and the certificates for a given input $I$ will depend on the machine used in verification.

The main result of this section is that all such counting functions
have additive approximation schemes under the `natural normalisation'
associated with the corresponding NDTM. For example if we take $f(G)$ to be the
number of Hamiltonian circuits in a graph $G$, then two possible NDTM's for checking
membership of $L$ are $M_1$ which takes
as certificates subsets of the edges, and checks that these form a
cycle of length $|V|$, and $M_2$ which takes as certificates an ordering of
the vertices $v_1,v_2,\ldots,v_n$ and checks that the edges between any two adjacent
vertices in the ordering, and between $v_n$ and $v_1$, do appear in the graph, (to avoid double
counting we must
insist that relative to some fixed ordering of the vertices
$\pi:V\mapsto 1,\ldots,n$, we have $\pi(v_1)=1$ and $\pi(v_2)<\pi(v_n)$). In each case the number of good certificates
for a given graph $G$ is exactly the number of Hamiltonian circuits of
$G$, however $M_1$ has $2^{|E(G)|}$ possible certificates,
while $M_2$ has $(|V|-1)!/2$ possible certificates. In either case the
number of possible certificates is a natural upper bound on the number of
Hamiltonian circuits. We show below that there is an additive
approximation algorithm under the normalisation associated with any such bound.

\begin{theorem} \label{sample} Let $f$ be a function  in the class
$\#P$, with an associated NDTM $M$, so that for a given instance $I$,
$M$ has $f(I)$ accepting certificates, each of length $p(|I|)$. Then
there exists an additive approximation algorithm for $(f,2^{p(|I|)})$ that runs in time polynomial in $|I|$ and $ \epsilon^{-1}$. \end{theorem}

\begin{proof} Given an instance $I$ of $f$, we will select $t$ computation paths, or certificates, uniformly at random from the $2^{p(|I|)}$ possible. We then run $M$ using these inputs, and let $X_i, i=1\ldots t$ be indicator functions which take value 1 if and only if the $i^{th}$ computation path accepts $I$. The estimator for $f(I)$ is then $X=\frac{2^{p(|I|)}}{t}\sum_{i=1}^t X_i $.

Clearly $\textbf{E}[X]=f(I)$. It remains to show that we can select $t$ only polynomially large, such that the error bounds given in Definition~\ref{AAdef} are satisfied. First note from Chebyshev's inequality that
\begin{eqnarray*} \mathbf{Pr}\big[|X-f(I)| \geqslant \epsilon2^{p(|I|)} \big] &\leqslant& \frac{\mathbf{Var}(X)}{\epsilon^2 2^{2p(|I|)}}\\
&\leqslant& \frac{1}{t} \frac{ \frac{f(I)}{2^{p(|I|)}}(1-\frac{f(I)}{2^{p(|I|)}})}{\epsilon^2 }\\
&\leqslant& \frac{1}{t\epsilon^2}. \end{eqnarray*}
Now if  $t=4\epsilon^{-2}+1$, we have
$$\mathbf{Pr}\big[|X-f(I)| \geqslant \epsilon2^{p(|I|)} \big]  < 1/4.$$
\end{proof}

Turning briefly to GapP functions, we get the following immediate corollary.
\begin{theorem} \label{Gapsample} Let $f$ be a GapP function such that $f=g-h$ where $ g,h\in$\#P.

\begin{enumerate}

\item Suppose that there are additive approximations for $(g,u)$ and $(h,v)$, then there is an additive approximation scheme for $(f,\max \{u,v\})$;
\item Suppose that $g(I)$ and $h(I)$ have certificates of length $p(|I|)$ for all $I$, then there is an additive approximation scheme for $(f,2^p)$. \end{enumerate}
\end{theorem}
\begin{proof}  From Proposition~\ref{add} we have that there is an additive approximation for $(f,u+v)$. From Definition~\ref{AAdef} we can halve the permitted error for only a polynomial increase in running time, hence there is an AA for $(f,\frac{u+v}2)$ and therefore also  $(f,\max \{u,v\})$, which gives (i). When $g$ and $h$ have certificates of length $p$,  by Theorem~\ref{sample} there are AA schemes for $(g,2^p)$ and $(h,2^p)$, (ii) now follows from (i).\end{proof}

 We have seen that all functions contained in \#P have an AA algorithm relative to normalisation by the size of the certificate space, and it
 is reasonable to ask if we could do better. However, we give two results that suggest this is already the best we can do in general. First we will see that sharpening our approximation to a logarithmic
 scale for the number of proper 5-colourings of a graph is NP-hard. Secondly, we show that the normalisation by the number of possible certificates cannot be improved significantly in the case of the number of $k$-colourings of a graph.

\begin{theorem}\label{log}
Let $P_G(5)$ be the number of proper 5-colourings of a graph. Then
for a general graph $G$ on $n$ vertices, there cannot be an additive approximation algorithm for $(\log (P_G(5)+1),3n)$ that runs in time polynomial in $n$ and $ \epsilon^{-1}$ unless  NP=RP.
\end{theorem}
\begin{proof}  Consider a NDTM for $P_G(5)$ that takes as certificates any 5-colouring of the graph, hence the certificates are of length $n\lceil \log 5 \rceil=3n$, where $n$ is the number of vertices in $G$. We show that an AA algorithm for $(\log (P_G(5)+1),3n)$ would be able to solve the NP-complete problem of determining whether a graph is 5-colourable.

Let $G$ be a graph on $n$ vertices and consider the following polynomial time transformation. We form $G^+$ by adding $n$ isolated vertices to $G$. If $G$ is not 5-colourable, then nor is $G^+$. However, if $G$ is 5-colourable, each 5-colouring can be extended to $5^{n}$ 5-colourings of $G^+$. Therefore, if $G$ is not 5-colourable $\log (P_{G^+}(5)+1)=0.$ Whereas, if $G$ is 5-colourable \begin{eqnarray*}
 \log (P_{G^+}(5)+1) &\geq& \log (5^{n}.5!+1)\\
&>& 2n.\end{eqnarray*}
Hence an additive approximation algorithm for $(\log (P_{G^+}(5)+1),6n)$ could determine whether $G$ is 5-colourable or not in random polynomial time. \end{proof}
Theorem~\ref{sample} shows that for connected graphs there exists an AA algorithm for $(P_G(k),(k-1)^n)$ as follows. We can take an arbitrary spanning tree on $G$, and take the set of certificates to define colourings relative to this spanning tree, giving $k(k-1)^{n-1}$ possible certificates. We can then adjust the normalisation by a constant factor $(k-1)/k$. We show that this cannot be improved, in the sense that the normalisation (and therefore the error) cannot be reduced by any exponential factor. We have already noted that the normalisation can be improved by any fixed polynomial factor.
\begin{theorem} \label{bestnormcol}
If NP $\neq$ RP then for any fixed $k \geq 3, \delta>0$ there cannot be a polynomial time AA algorithm for $(P_{G}(k),\phi(n))$ for connected graphs $G$ on $n$ vertices, for any function $\phi(n)$ of order $O((k-1-\delta)^n)$.
\end{theorem}
\begin{proof}  Let $\phi(n) \leqslant c (k-1-\delta)^n$ for sufficiently large $n$. Take $r$ such that $(k-1-\delta)\leq (k-1)^{1-1/r}.$ Given any graph $G$, we form a graph $H$ by attaching a path of length $n(r-1)$ to a vertex of $G$. Now \begin{eqnarray}P_G(k) (k-1)^{n(r-1)}&=&P_H(k)\\
P_G(k)&=&\frac{P_H(k)}{((k-1)^{1-1/r})^{nr}}\label{gequalsh}\\
P_G(k)&=&\frac{c P_H(k)}{\phi(nr)}\frac{\phi(nr)}{c ((k-1)^{1-1/r})^{nr}}.\label{philesshalf}\end{eqnarray}
Now suppose that there is an AA algorithm for  $(P_H(k),\phi(|H|))$, then we can get an approximation $\hat{P_H}(k)$  within an additive error of $\frac{1}{2c} \phi(nr)$. Using $\hat{P_H}(k)$ and equation~(\ref{gequalsh}), we obtain an approximation $\hat{P_G}(k)$. By equation~(\ref{philesshalf}), $\hat{P_G}(k)$ is within an additive error of 1/2, since  $$\frac{\phi(nr)}{c ((k-1)^{1-1/r})^{nr}} \leqslant \frac{\phi(nr)}{c (k-1-\delta)^{nr}} \leqslant 1.$$
Since $P_G(k)$ is integral, we can therefore determine it exactly.
\end{proof}


\section{Approximating $V_L(t)$ and related quantities}\label{TP}
We have seen in Section~\ref{jones} that our primary problem is to decide whether or not there exists an additive approximation scheme for $(V_L(e^{2\pi i/5}),[2]_5^{m/2})$ where $L$ is the plat closure of a braid on $m$ strings, indeed an additive approximation for the absolute value of the Jones polynomial suffices. We make this precise in the following theorem.

\begin{theorem}\label{bqpeqpa} Let $A$ be a oracle which takes as input a braid $b$ on $m$ strings and $\epsilon>0$, and returns an additive approximation for $(|V_L(e^{2\pi i /5})|,[2]_5^{m/2})$, where $L$ is the plat closure of $b$. Then $\mathrm{BQP} = \mathrm{P}^A.$ \end{theorem}

\begin{proof}  Recall from Section~\ref{jones} that given a braid on $m$ strings, a topological quantum computer can be constructed such that the probability of output zero is $\frac{|V_L(e^{2\pi i /5})|^2}{[2]_5^m}$ and the computer runs in time polynomial in $m$. Given $\epsilon>0$, using independent runs of this computer and a standard sampling approach, the probability of zero can be estimated to within an error of ${\epsilon^2}$, where the number of runs is polynomial in $\epsilon^{-1}$. Hence we may estimate ${|V_L(e^{2\pi i /5})|}$ to within an absolute error of ${\epsilon}{[2]_5^{m/2}}$ in polynomial time. Therefore $\mathrm{P}^A \subseteq \mathrm{BQP}$.

Secondly, suppose we have a BQP language and an input $x$. By Theorem~\ref{FKLWtheorem} we can determine a link $L$, of size polynomial in $|x|$, such that $L$  satisfies Equation~(\ref{FKLWeqn1}), and the number of minima of $L$ is $|x|+1$.  If $x$ is in the language, $\textbf{Pr}(0)<1/4$, hence
\begin{eqnarray} 0\leqslant \frac{1}{1+[2]_5^2} \bigg( 1+
\frac{(-1)^{c(L)+w(L)}(-a)^{3w(L)} V_L(e^{2\pi i/5})}{
[2]_5^{|x|-1}}\bigg) &<& 1/4 \nonumber \\
-[2]_5^{|x|+1}[2]_5^{-2} \leqslant (-1)^{c(L)+w(L)}(-a)^{3w(L)} V_L(e^{2\pi i/5}) &<& [2]_5^{|x|+1}[2]_5^{-2} \left( \frac{1+[2]_5^2}{4} -1\right)\nonumber \\
|V_L(e^{2\pi i/5})| &<& [2]_5^{|x|+1} 0.39. \end{eqnarray}
Whereas, if $x$ is not in the language, $\textbf{Pr}(0)>3/4$, hence
\begin{eqnarray} \frac{1}{1+[2]_5^2} \bigg( 1+
\frac{(-1)^{c(L)+w(L)}(-a)^{3w(L)} V_L(e^{2\pi i/5})}{
[2]_5^{|x|-1}}\bigg) &>& 3/4 \nonumber \\
(-1)^{c(L)+w(L)}(-a)^{3w(L)} V_L(e^{2\pi i/5}) &>& [2]_5^{|x|+1}[2]_5^{-2} \left( \frac{3(1+[2]_5^2)}{4} -1\right)\nonumber \\
|V_L(e^{2\pi i/5})| &>& [2]_5^{|x|+1} 0.65. \end{eqnarray}
Clearly use of an oracle giving an additive approximation for $(|V_L(e^{2\pi i /5})|,[2]_5^{|x|+1})$ will enable us to distinguish these two cases with probability at least 3/4. Hence $\mathrm{BQP} \subseteq \mathrm{P}^A$.

\end{proof}

We saw in Section~\ref{jones} the equivalence of the Jones polynomial and a specialisation of the Tutte polynomial for alternating links, hence we would like an AA for a general planar graph $G$ for $$(T(G;-e^{2\pi i/5},-e^{-2\pi i/5}),u)$$ where $u$ is some reasonable upper bound.

Hyperbolae of the form $H_q :=(x-1)(y-1)=q$ play a crucial role in the manipulation of the Tutte polynomial; loosely speaking the process of performing a \emph{tensor product} on an input graph $G$ with some other fixed graph $N$ enables us to `move around' the Tutte plane. That is the new graph $G\otimes N$ satisfies: \begin{eqnarray}\label{hyp}
T(G\otimes N;x,y)=f(N;x,y)T(G;X,Y)
\end{eqnarray}
where $f$ and the arguments $X,Y$ can be computed in time polynomial in $|x|, |y|$ and $|N|$. However, for any choice of $N$, the new points $X,Y$ satisfy, \begin{eqnarray}
(X-1)(Y-1)=(x-1)(y-1)=q.
\end{eqnarray} Thus, such transformations restrict us to remain on the initial hyperbola $H_q$; see~\cite{jae90} for further details. The close relationship between points on the hyperbola enables us to use an additive approximation at one point to get an additive approximation at any other point $(X,Y)$ on the same hyperbola for which there exists a suitable planar $N$ which transforms $(x,y)$ to $(X,Y)$ by~(\ref{hyp}).
\begin{proposition}
Let $x,y\in \mathbb{Q}$ and $N$ a planar graph on k $vertices$ be fixed. Suppose there is an AA sceme for $(T(G;x,y),u(n))$ for any planar $G$ on $n$ vertices and $m$ edges. Then there is also an AA for $$\bigg(T(G;X,Y),u(n+m(k-2))\bigg),$$
where $X$ and $Y$ are the points determined by the transformation in~(\ref{hyp}) (depending only on $x,y$ and $N$).
\end{proposition}
\begin{proof} Let $X$ and $Y$ be the points satisfying~(\ref{hyp}). Since $G\otimes N$ is planar (see the construction in~\cite{jae90}) we may use the AA scheme for $$(T(G\otimes N;x,y),u(|V(G\otimes N)|))$$ to get an approximation to within an error $\epsilon u(|V(G\otimes N)|)$ with probability at least $3/4$ in polynomial time. Note that $|V(G\otimes N)| =n+m(k-2)$. Since the running time of the AA scheme is polynomial in $\epsilon^{-1}$, and $f(N;x,y)$ is a constant, we can approximate to within an error of $\epsilon |f(N;x,y)|u(n+m(k-2))$ and still run in polynomial time. By~(\ref{hyp}) we have
$$T(G;X,Y)=\frac{T(G\otimes N;x,y)}{f(N;x,y)},$$
Hence the AA for $T(G\otimes N;x,y)$ yields an AA for $T(G;X,Y)$ with error at most $\epsilon u(n+m(k-2))$ in polynomial time. \end{proof}

Because of the important role of these hyperbolae, it is natural to look at the hyperbolae containing the roots of unity $(-e^{\frac{2\pi i}{n}},-e^{-\frac{2\pi i}{n}})$. These are $H_{q_n}$, where $q_n=2+2\cos (2\pi/n)$, which cut the x-axis at \begin{eqnarray}
x=-1-2\cos (2\pi/n),
\end{eqnarray} corresponding to an evaluation of the chromatic polynomial at one of the well known \emph{Beraha numbers} $B_n=2+2\cos (2\pi/n)$. Since for real $x$ and $y$ and any graph $N$, the related points $X$ and $Y$ will also be real, we cannot find an $N$ such that we can directly relate $T(G\otimes N;1-B_5,0)$ and $T(G;-e^{\frac{2\pi i}{5}},-e^{-\frac{2\pi i}{5}})$. Whether or not we can find a point within absolute value $\epsilon$ of $(1-B_5,0)$ that can be directly related to $(-e^{\frac{2\pi i}{5}},-e^{-\frac{2\pi i}{5}})$ is an interesting ongoing question.
We present some positive results below, and return to these
difficulties in Section~\ref{openqs}.

First note that by Theorems~\ref{sample} and~\ref{Gapsample} we know that $T(G;x,y),\ x,y\in\mathbb{Z}$ will have an AA scheme with respect to an appropriate normalisation, since evaluations at these points are GapP functions. However the drawback is that often the naive normalisation will be too large. This will not always be the case, for example the point $(1-\lambda,0), \lambda\in \mathbb{Z},$ gives the number of proper $\lambda$ colourings, and by Theorems~\ref{log} and~\ref{bestnormcol} here we have a best possible normalisation.

When we consider the non-integer points, the situation is more complicated. A straight forward sampling approach gives the following result.

\begin{proposition}
For rational $(x,y)$ and a connected graph $G$, there exists a AA algorithm for the following:
\begin{enumerate}
\item $(T(G;x,y),y^{|E|}(y-1)^{-|V|+1})$ when  $\{x=1,y>1\}$,
\item $(T(G;x,y),x^{|E|}(x-1)^{|V|-|E|-1})$ when  $\{x>1,y=1\}$,
\item $(T(G;x,y),y^{|E|}(x-1)^{|V|-1})$ when  $\{x> 1,y> 1\}$.
\end{enumerate}
\end{proposition}

In other regions, in particular where there are negative terms in the expansion of $T$, cancellation between terms means that there is no longer a natural upper bound by which to normalise. We return to this problem in Section~\ref{openqs}.


\section{An alternative approach}\label{alternative}

Returning to our original motivation, at the moment we are unable to determine whether there is an AA algorithm for $(V_L(e^{2\pi i/5}),[2]_5^{m/2})$ where $L$ is the plat closure of a braid on $m$ strings. However we now show that in order to simulate a quantum computation it would be sufficient to determine the sign of the real part of $V_L(e^{2\pi i/5})$. Particularly in the case that the writhe of $L$ is zero, and hence $V_L(e^{2\pi i/5})$ is real, this appears an easier problem.
\begin{theorem}
Let $A(L)$ be an oracle that returns the sign of the real part of the Jones polynomial of the link $L$ evaluated at $e^{2\pi i/5}$. Then $\mathrm{BQP}\subseteq \mathrm{P}^A.$
\end{theorem}

\begin{proof}
This proof follows that of Theorem~\ref{bqpeqpa}. Suppose we have a BQP language and an input $x$.
By Theorem~\ref{FKLWtheorem} we can determine a link $L$, of size polynomial in $|x|$, such that $L$  satisfies Equation~(\ref{FKLWeqn1}), and the number of minima of $L$ is $|x|+1$.   We now assume that $w(L)=0$; the proof in the cases $w(L)=4$ and $w(L)=-4$ follow by a similar argument. Simplifying Equation~(\ref{FKLWeqn1}) in the case $w(L)=0$ we have
\begin{eqnarray}
 \mathbf{Pr}(0)&=&\frac{1}{1+[2]_5^2} \bigg( 1+ \frac{(-1)^{c(L)} V_L(e^{2\pi i/5})}{ [2]_5^{m(L)-2}}\bigg).\end{eqnarray}
If $x$ is in the language, $\textbf{Pr}(0)<1/4$, hence
\begin{eqnarray} \frac{1}{1+[2]_5^2} \bigg( 1+
\frac{(-1)^{c(L)} V_L(e^{2\pi i/5})}{
[2]_5^{|x|-1}}\bigg) &<& 1/4 \nonumber \\
(-1)^{c(L)} V_L(e^{2\pi i/5}) &<& [2]_5^{|x|-1} \left( \frac{1+[2]_5^2}{4} -1\right)\nonumber \\
(-1)^{c(L)} V_L(e^{2\pi i/5}) &<& -[2]_5^{|x|-1} 0.09<0. \end{eqnarray}
Whereas, if $x$ is not in the language, $\textbf{Pr}(0)>3/4$, hence
\begin{eqnarray} \frac{1}{1+[2]_5^2} \bigg( 1+
\frac{(-1)^{c(L)} V_L(e^{2\pi i/5})}{
[2]_5^{|x|-1}}\bigg) &>& 3/4 \nonumber \\
(-1)^{c(L)} V_L(e^{2\pi i/5}) &>& [2]_5^{|x|-1} \left( \frac{3(1+[2]_5^2)}{4} -1\right)\nonumber \\
(-1)^{c(L)} V_L(e^{2\pi i/5}) &>& [2]_5^{|x|-1} 1.71>0. \end{eqnarray}
Clearly use of an oracle giving an additive approximation for the sign of $V_L(e^{2\pi i /5})$ will enable us to distinguish these two cases with probability at least 3/4. Hence $\mathrm{BQP} \subseteq \mathrm{P}^A$.\end{proof}

In the previous section we outlined the importance of the hyperbolae $H_q:=(x-1)(y-1)=q$ to the Tutte polynomial. For $x,y,X,Y$ and $N$ related as in Equation~(\ref{hyp}), we can determine the sign of $T(G;X,Y)$ if we can determine the sign of $T(G\otimes N;x,y)$.

This gives rise to the natural question of the complexity of determining whether a function is greater than or less than zero, in particular the Tutte polynomial, of which the Jones is a specialization. It is immediate from the definitions that the Tutte is non-negative in the region $x,y\geqslant 0$. At all other integer points on the axes the Tutte polynomial counts either colourings or flows, up to easy multiplicative factors. Since these factors may be positive or negative, we can always select one of either ``$T(G;x,y)$ is non-negative'' or ``$T(G;x,y)$ is non-positive'' that is true, in polynomial time. In the above situation we are not concerned with cases in which the value is exactly zero, hence this would suffice. We consider the situation at other points in the next section.


\section{Some combinatorial and complexity questions}\label{openqs}

We close with the following questions
which have been prompted by this work.

In Section~\ref{TP} we noted that we are unable to find a suitable normalisation for approximating the Tutte polynomial  when the expansion included negative terms. We return to this here and examine the chromatic polynomial to highlight the difficulties. We have seen that for a connected graph $G$, we have an additive approximation for $(P_G(\lambda),(\lambda-1)^n))$ for all $\lambda\in \mathbb{Z}^+$. However we are most interested in an additive approximation at the non-integral Beraha numbers. One  might hope to achieve the above approximation for all $\lambda\in\mathbb{R}^{>1}$, however this seems unlikely as $(\lambda-1)^n$ is not even close to being an upper bound for $P_G(\lambda)$. Indeed consider the complete graphs: for small $\delta$
\begin{eqnarray}
P_{K_n}(1+\delta)&=&(1+\delta)(\delta)(-1+\delta)\cdots (-n+\delta+1)\\
&\approx&(-1)^{n-2} \delta (n-2)!. \end{eqnarray}
This prompts the first open question.
\begin{question} What is the best upper bound depending on $\lambda, n$ and $m$ for $|P_G(\lambda)|$ for all (planar) graphs $G$ on $n$ vertices and $m$ edges? \end{question}
As far as we are
aware the best upper bound known~\cite{woo77} is: \begin{displaymath} |P_G(\lambda)|\leqslant
|\lambda|^{n-m}(|\lambda|+1)^m\quad \lambda\in\mathbb{C}.\end{displaymath}
For general graphs we can make the following small improvement.
\begin{proposition} Let $G$ be a graph and let $\lambda \in \mathbb{C}$, then
\begin{eqnarray*}
|P_{G}(\lambda)|&\leqslant&\left(\frac mn -1\right)^{n-m}\left(\frac mn \right)^m\qquad \textrm{for }\frac mn \geqslant |\lambda|+1\\
|P_{G}(\lambda)|&\leqslant&(|\lambda|)^{n-m}(|\lambda|+1 )^m\qquad\quad  \textrm{for } \frac mn < |\lambda|+1. \end{eqnarray*}
If $G$ is a connected graph then
\begin{eqnarray*}
|P_{G}(\lambda)|&\leqslant&\left(\frac m{n-1} -1\right)^{n-m-1}\left(\frac m{n-1} \right)^m |\lambda| \quad \textrm{for } \frac m{n-1} \geqslant |\lambda-1|\\
|P_{G}(\lambda)|&\leqslant&(|\lambda-1| -1)^{n-m-1}(|\lambda-1|)^m |\lambda| \quad\quad \textrm{for } \frac m{n-1} < |\lambda-1|.\end{eqnarray*}\end{proposition}

These bounds hold for all $\lambda\in \mathbb{C}$, however the Beraha numbers have special characteristics.
The evaluations of the chromatic polynomial at these points have some beautiful, but not totally understood, properties~\cite{tut70}. The values begin $ 4, 0,1, 2 ,1+ \tau, 3, \ldots$ and converge towards 4, where $\tau$ is the golden ratio $\frac{1+\sqrt{5}}{2}$. The integers in this series are clearly central to the theory of chromatic polynomials. Writing $B_5=1+\tau$, then for any plane triangulation $T$ on $n$ vertices: \begin{eqnarray}
|P_T(B_5)|\leqslant\tau^{5-n}\label{gold1}.
\end{eqnarray}
For a connected graph $G$ with average degree at least $3.24$ (note that a planar triangulation has average degree $6-12/n$), the above proposition gives,
$$|P_{G}(B_5)|\leqslant\left(\frac m{n-1} -1\right)^{n-m-1}\left(\frac m{n-1} \right)^m B_5.$$
Hence we ask:
\begin{question} Is there a better bound for $|P_G(B_n)|$ than there is for an evaluation at a general point? \end{question}

Following the results of Section~\ref{alternative} we are also prompted to examine the complexity of determining whether the Tutte polynomial is greater than or equal to, or less than zero at a given point. Recall that this decision problem is trivial for $x,y\geqslant 0$, and for integer points on the axes. Again considering the specialization to the chromatic polynomial we ask:
\begin{question} For fixed $\lambda\in \mathbb{Q}$, is it NP-hard to decide whether $P_G(\lambda)$ is greater than or equal to, or less than zero?\end{question}

Note that this is trivial for $\lambda\in\mathbb{Z}$. It is also
P-time decidable for $\lambda < 32/27$ by the following theorem of
Woodall~\cite{woo77} and Jackson~\cite{jac93}.
\begin{theorem} Let $G$ be a graph without loops on $n$ vertices, $\kappa$ components and $b$ blocks.\begin{enumerate}
\item If $\ \lambda<0$, then $P_G(\lambda)$ is non-zero with the sign of $(-1)^{n}$;
\item If $\ 0<\lambda<1$, then $P_G(\lambda)$ is non-zero with the sign of $(-1)^{n-\kappa}$;
\item If $\ 1<\lambda<\frac{32}{27}$, then $P_G(\lambda)$ is non-zero with the sign of $(-1)^{n-\kappa-b}$.\end{enumerate}
\end{theorem}
Note that $P_G(\lambda)\neq 0$ for $\lambda\in\mathbb{Q}\backslash \mathbb{Z}$, since the chromatic polynomial has integer coefficients. It is easy to show the following.\begin{itemize}
\item Let $\lambda \in \mathbb{Q}\backslash \mathbb{Z}$. If deciding whether $P_G(\lambda)>0$ is NP-hard, then it is also NP-hard to decide whether $P_G(\lambda+1)>0$ for a general graph $G$. \end{itemize}
However since it is easy to decide for $\lambda<32/27$, the converse cannot be true for all $\lambda \in \mathbb{Q}\backslash \mathbb{Z}$ unless these questions are all in P. It would be interesting to know the answer to the following questions.
\begin{question} Does there exist a critical $\alpha>0$ such that deciding whether $P_G(\lambda)$ is greater than or less than zero is NP-hard for all rational $\lambda>\alpha$, $\lambda \not\in \mathbb{Z}$?\end{question}
\begin{question} Is this critical $\alpha$ equal to $32/27$?\end{question}

As before we are more interested in evaluating the chromatic polynomial at the Beraha points than at general non-integers, and the graphs we are most interested in are planar. Hence we ask the specific question:
\begin{question} For planar graphs, is the problem of deciding whether $P_G(B_n)$ is greater or less than zero NP-hard?\end{question}
For any graph $G$, not necessarily planar, it is known that $P_G(B_n)\neq 0$ for $n\geq 5, n\neq 6,10$, \cite{sal00}.
Also Tutte~\cite{tut70} has shown that for any planar triangulation the following equation holds, writing $B_{10} =\tau \sqrt{5}$,  \begin{eqnarray}\label{golden}
P_T(B_{10})=\sqrt{5} \tau^{3(n-3)} (P_T(B_5))^2.
\end{eqnarray}
So $P_T(B_{10})>0$ for all plane triangulations $T$, indeed a simple  reverse induction shows that for any planar graph $G$, $P_G(B_{10})>0$ holds. Further, for any outerplanar graph $G$, we can form the planar graph $G^+$ by adding a new vertex adjacent to all original vertices. Since
$$P_{G^+}(\lambda+1)=(\lambda+1) P_G(\lambda)$$
holds for all positive integers, it holds for all $\lambda \in \mathbb{R}$. Noting that $B_{10}=B_5+1$, we conclude that $P_G(B_5)>0$ for all outerplanar graphs $G$.

\section*{Acknowledgements}
The authors would like to thank the referee for many helpful comments and suggestions, in particular for suggesting a simpler proof of Theorem~\ref{log}. We also thank Graham Brightwell and Colin McDiarmid for comments on an earlier version of this work~\cite{bor03c}.

\end{document}